\documentclass[12pt,a4paper]{article}
\usepackage[english]{babel}
\usepackage[top=1in, bottom=1.25in, left=1.25in, right=1.25in]{geometry}
\usepackage[hidelinks]{hyperref}
\usepackage{amsmath}
\usepackage{amsthm}
\usepackage{amssymb}
\usepackage[ruled,vlined,linesnumbered]{algorithm2e}
\usepackage{tikz}

\newcommand{\ttt}[1]{\texttt{#1}}
\newcommand{\pp}[2]{p_{#1}[#2]}
\newcommand{\CH}{\mathcal{H}}
\newcommand{\NULL}{\ttt{NULL}}
\newcommand{\RHP}{\textrm{RHP}}
\newcommand{\LHP}{\textrm{LHP}}
\newcommand{\LLL}{\mathcal L}
\newcommand{\TTT}{\mathcal T}
\newcommand{\comm}[1]{}

\newtheorem{theorem}{Theorem}
\newtheorem{observation}[theorem]{Observation}
\newtheorem{lemma}[theorem]{Lemma}

\title{An Optimal Algorithm for the Separating Common Tangents of two
Polygons\footnote{A preliminary version of this paper appeared at
SoCG 2015 \cite{abrahamsen}.
In the case where the convex hulls of the polygons are not disjoint,
it is not clear that the algorithm for separating common tangents terminates
within the given bound on the running time.
Here, we give a correct algorithm and simplify the proof of correctness slightly.}}

\author{Mikkel Abrahamsen\thanks{Research partly supported by Mikkel Thorup's
    Advanced Grant from the Danish Council for Independent Research
    under the Sapere Aude research career programme.}\\
    Department of Computer Science, University of Copenhagen\\
  Universitetsparken 5\\
  DK-2100 Copanhagen \O\\
  Denmark\\
  \texttt{miab@di.ku.dk}
   }

\newtheorem*{theorem*}{Theorem}

\begin{document}

\maketitle

\begin{abstract}
We describe an algorithm for computing the separating common tangents of two simple polygons
using linear time and only constant workspace.
A tangent of a polygon is a line touching the polygon such that all of the polygon lies
to the same side of the line. A separating common tangent of two polygons is a tangent
of both polygons where the polygons are lying on different sides of the tangent.
Each polygon is given as a read-only array of its corners.
If a separating common tangent does not exist, the algorithm reports that.
Otherwise, two corners defining a separating common tangent are returned.
The algorithm is simple and implies an optimal
algorithm for deciding if the convex hulls of two polygons are disjoint or not.
This was not known to be possible in linear time and constant workspace
prior to this paper.

An outer common tangent is a tangent of both polygons where the polygons are
on the same side of the tangent.
In the case where the convex hulls of the polygons are disjoint, we give an algorithm for computing
the outer common tangents in linear time using constant workspace.
\end{abstract}

\section{Introduction}

The problem of computing common tangents of two given polygons has received some attention
in the case where the polygons are convex. For instance, it is necessary to compute outer common
tangents of disjoint convex polygons in the classic divide-and-conquer algorithm
for the convex hull of a
set of $n$ points in the plane by Preparata and Hong \cite{preparata1977}.
They give a naïve linear time algorithm
for outer common tangents since that suffices for an $O(n\log n)$ time convex hull algorithm.
The problem is also considered in various dynamic convex hull algorithms
\cite{brodal2002, hershberger1992, overmars1981}.
Overmars and van Leeuwen \cite{overmars1981} give an $O(\log n)$ time algorithm for computing
an outer common tangent of two disjoint convex polygons when a separating line is known, where
each polygon has at most $n$ corners.
Kirkpatrick and Snoeyink \cite{kirkpatrick19952} give an $O(\log n)$ time algorithm for the same problem, but without using a separating line.
Guibas et al.~\cite{guibas1991} give an $\Omega(\log^2 n)$ lower bound on the time required
to compute an outer common tangent of two intersecting convex polygons,
even if it is known that they intersect in at most two points.
They also describe an algorithm achieving that bound.

Touissaint \cite{toussaint1983} considers the problem of computing separating common
tangents of convex polygons and notes that the problem occurs in problems related to
visibility, collision avoidance, range fitting, etc. He gives a linear time algorithm. 
Guibas et al.~\cite{guibas1991} give an $O(\log n)$ time algorithm for the same problem.

All the here mentioned works make use of the convexity of the polygons.
If the polygons are not convex, one can use a linear time algorithm to compute
the convex hulls before computing the tangents \cite{graham1983, melkman1987}.
However, if the polygons are given in read-only memory,
it requires $\Omega(n)$ extra bits to store the convex hulls.
In this paper, we also
obtain linear time while using only constant workspace, i.e.~$O(\log n)$ bits.
For the outer common tangents, we require the convex hulls of the polygons
to be disjoint.
There has been some recent interest in constant workspace algorithms for
geometric problems, see for instance \cite{abrahamsen2013, asano1, asano2, barba2}.

The problem of computing separating common tangents is of special interest because
these only exist when the convex hulls of the polygons are disjoint, and our
algorithm detects if they are not. Thus, we also provide an optimal algorithm for deciding
if the convex hulls of two polygons are disjoint or not. This was to the best of our knowledge
not known to be possible in linear time and constant workspace prior to our work.

\subsection{Notation and some basic definitions}

Given two points $a$ and $b$ in the plane, the closed line segment with endpoints $a$ and $b$
is written $ab$.
When $a\neq b$, the line containing $a$ and $b$ which is infinite in both directions
is written $\LLL(a,b)$.

Define the dot product of two points $x=(x_0,x_1)$ and $y=(y_0,y_1)$ as
$x\cdot y=x_0y_0+x_1y_1$, and let
$x^{\perp}=(-x_1,x_0)$ be the counterclockwise rotation of $x$ by the angle $\pi/2$.
Now, for three points $a$, $b$, and $c$, we define
$\TTT(a,b,c)=\textrm{sgn}((b-a)^{\perp}\cdot(c-b))$,
where $\textrm{sgn}$ is the sign function.
$\TTT(a,b,c)$ is $1$ if $c$ is to the left of the directed line from $a$ to $b$,
$0$ if $a$, $b$, and $c$ are collinear, and $-1$ if $c$ is to the right of the directed line
from $a$ to $b$. We see that
$$\TTT(a,b,c)=\TTT(b,c,a)=\TTT(c,a,b)=-\TTT(c,b,a)=-\TTT(b,a,c)=-\TTT(a,c,b).$$
We also note that if $a'$ and $b'$ are on the
line $\LLL(a,b)$ and appear in the same order as $a$ and $b$, i.e., $(b-a)\cdot (b'-a') > 0$,
then $\TTT(a,b,c)=\TTT(a',b',c)$ for every point $c$.

The \emph{left half-plane} $\LHP (a,b)$ is the closed half plane with boundary $\LLL(a,b)$
lying to the left of directed line from $a$ to $b$, i.e., all the points $c$ such that
$\TTT(a,b,c)\geq 0$.
The \emph{right half-plane} $\RHP(a,b)$ is just $\LHP(b,a)$.

Assume for the rest of this paper that
$P_0$ and $P_1$ are two simple polygons in the plane with $n_0$ and $n_1$ corners, respectively,
where $P_k$ is defined
by its corners $\pp k0,\pp k1,\ldots,\pp k{n_k-1}$ in clockwise or counterclockwise order,
$k=0,1$.
Indices of the corners are considered modulo $n_k$, so that $\pp ki$ and
$\pp kj$ are the same corner when $i\equiv j\pmod {n_k}$.

We assume that the corners are in general position in the sense
that $P_0$ and $P_1$ have no common corners and
the union of corners
$\bigcup_{k=0,1}\{\pp k0,\ldots,\pp k{n_k-1}\}$ contains no three collinear corners.

A \emph{tangent} of $P_k$ is a line $\ell$ such that $\ell$ and $P_k$ are not disjoint and such that
$P_k$ is contained in one of the closed half-planes defined by $\ell$. The line $\ell$ is a \emph{common tangent} of $P_0$ and $P_1$ if it is a tangent of both $P_0$ and $P_1$.
A common tangent is an \emph{outer} common tangent if $P_0$ and $P_1$ are on the same side of
the tangent, and otherwise the tangent is \emph{separating}. See Figure \ref{commonTangentsEx}.

\begin{figure}
\centering
\begin{tikzpicture}[scale=0.2]

\input{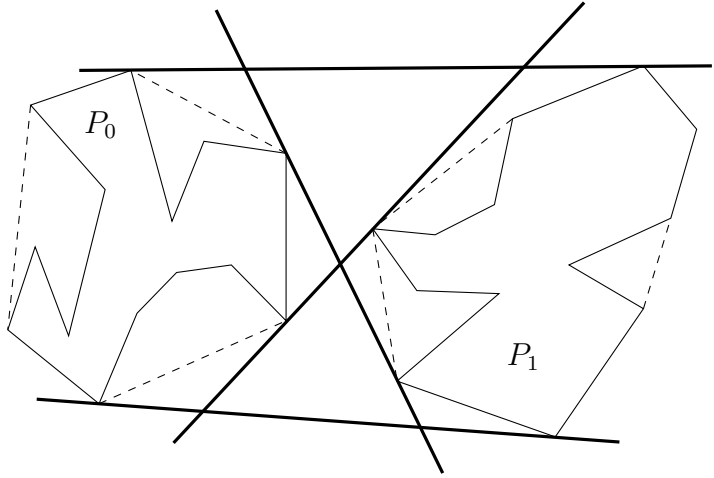}

\draw [dashed] (Z) -- (C);
\draw [dashed] (D) -- (H);
\draw [dashed] (I) -- (Y);

\draw [dashed] (W) -- (M);
\draw [dashed] (O) -- (Q);
\draw [dashed] (T) -- (W);

\draw [very thick] (5.8, 41.0) -- (47.4, 41.3);
\draw [very thick] (3.0, 19.2) -- (41.32, 16.36);
\draw [very thick] (14.8, 45.0) -- (29.7, 14.3);
\draw [very thick] (12.0, 16.3) -- (39.0, 45.5);




\end{tikzpicture}
\caption{Two polygons $P_0$ and $P_1$
and their four common tangents as thick lines. The edges of the convex
hulls which are not edges of $P_0$ or $P_1$ are dashed.}
\label{commonTangentsEx}
\end{figure}

For a simple polygon $P$, we let $\CH(P)$ be the convex hull of $P$.
The following lemma is a well-known fact about $\CH(P)$.

\begin{lemma}
For a simple polygon $P$, $\CH(P)$ is a convex polygon
and the corners of $\CH(P)$ appear in the same
cyclic order as they do on $P$.
\end{lemma}

The following lemma states folklore properties of tangents of polygons.

\begin{lemma}\label{folklore}
A line is a tangent of a polygon $P$ if and only if it is a tangent of $\CH(P)$.
Under our general position assumptions, the following holds:
If one of $\CH(P_0)$ and $\CH(P_1)$ is completely contained
in the other, there are no outer common tangents of $P_0$ and $P_1$.
Otherwise, there are two or more. There are exactly two if $P_0$ and $P_1$ are disjoint.
If $\CH(P_0)$ and $\CH(P_1)$ are not disjoint, there are no separating common
tangents of $P_0$ and $P_1$. Otherwise, there are exactly two.
\end{lemma}

\section{Computing separating common tangents}\label{sepTan}

In this section, we assume that the corners of $P_0$ and $P_1$ are both
given in counterclockwise order.
We prove that Algorithm \ref{sepTanAlg} returns a pair of indices
$(s_0,s_1)$ such that the line
$\LLL(\pp 0{s_0},\pp 1{s_1})$ is a separating common tangent with
$P_k$ contained in $\RHP(\pp{1-k}{s_{1-k}},\pp k{s_k})$ for $k=0,1$.
If the tangent does not exist, the algorithm returns $\NULL$.
The other separating common tangent can be found by a similar algorithm if the corners of the polygons
are given in clockwise order and `$=1$' is changed to `$=-1$' in line \ref{testSide}.

\begin{algorithm}[H]
\LinesNumbered
\DontPrintSemicolon
\SetArgSty{}
\SetKwInput{Input}{Input}
\SetKwInput{Output}{Output}
\SetKw{Report}{report}
\SetKwIF{If}{ElseIf}{Else}{if}{}{else if}{else}{end if}
\SetKwFor{Foreach}{for each}{}{end for}
\SetKwFor{While}{while}{}{end while}
$s_0\gets 0$;\quad$t_0\gets 1$;\quad$s_1\gets 0$;\quad$t_1\gets 1$;\quad$u\gets 0$\;
\nllabel{initialize}
\While{$t_0<3n_0$ or $t_1<3n_1$} {\nllabel{while}
  \If {$\TTT(\pp {1-u}{s_{1-u}},\pp u{s_u},\pp u{t_u})=1$} {\nllabel{testSide}
    \If {$t_u\geq 2n_u$} {\nllabel{testStop}
      \Return {$\texttt{NULL}$}\;\nllabel{testStopAct}
    }  
    $s_u\gets t_u$\; \nllabel{updateS}
    $t_{1-u}\gets s_{1-u}+1$\;
  }
  $t_u\gets t_u+1$\;
  $u\gets 1-u$\;
}
\Return {$(s_0,s_1)$}
\caption{$\ttt{SeparatingCommonTangent}(P_0,P_1)$}
\label{sepTanAlg}
\end{algorithm}

\begin{figure}
\centering
\begin{tikzpicture}[scale=0.2]

\input{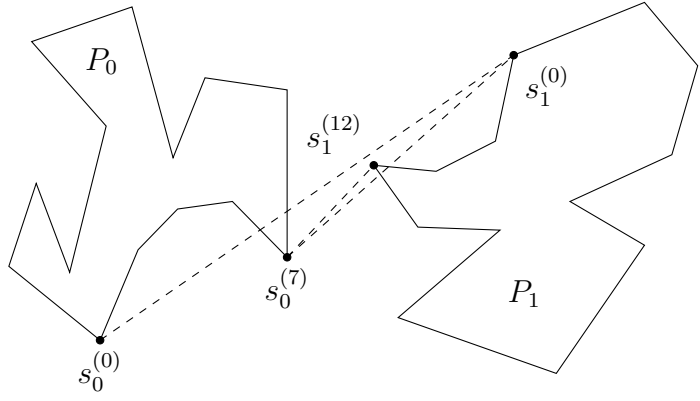}

\draw [dashed] (I) -- (T);
\draw [dashed] (Y) -- (T);
\draw [dashed] (Y) -- (W);

\filldraw [black] (I) circle (7pt);
\node at (I) [below] {$s_0^{(0)}$};
\filldraw [black] (T) circle (7pt);
\node at (T) [anchor=north west] {$s_1^{(0)}$};
\filldraw [black] (Y) circle (7pt);
\node at (Y) [below] {$s_0^{(7)}$};
\filldraw [black] (W) circle (7pt);
\node at (W) [anchor = south east] {$s_1^{(12)}$};



\end{tikzpicture}
\caption{Algorithm \ref{sepTanAlg} running on two polygons $P_0$ and $P_1$.
The corners $\pp k{s_k^{(i)}}$ are marked and labeled as
$s_k^{(i)}$ for the initial values
$s_k^{(0)}$ and after each iteration $i$ where an update of $s_k$ happens.
The segments $\pp 0{s_0^{(i)}}\pp 1{s_1^{(i)}}$
on the temporary line are dashed.}
\label{algRunningEx}
\end{figure}

The algorithm traverses the polygons in parallel one corner at a time
using the indices $t_0$ and $t_1$.
We say that the indices $(s_0,s_1)$
define a \emph{temporary line}, which is the line $\LLL(\pp 0{s_0},\pp 1{s_1})$.
We update the indices $s_0$ and $s_1$ until the temporary line is
the separating common tangent.
At the beginning of an iteration of the loop at line \ref{while}, we traverse one corner
$\pp u{t_u}$ of $P_u$, $u=0,1$. If the corner happens to be on the wrong side of the intermediate
line, we make the temporary line pass through that corner by
updating $s_u$ to $t_u$ and we reset $t_{1-u}$ to $s_{1-u}+1$. The reason for
resetting $t_{1-u}$ is that a
corner of $P_{1-u}$ which was on the correct side of the old temporary line can
be on the wrong side of the new line and thus needs be traversed again.

We show that if the temporary line is not a separating common tangent after each polygon
has been traversed twice,
then the convex hulls of the polygons are not disjoint.
Therefore, if a corner is found to be on the wrong side of the temporary line
when a polygon is traversed for the third time,
no separating common tangent can exist and
$\NULL$ is returned.
Let $s_k^{(i)}$ be the value of $s_k$ after $i=0,1,\ldots$ iterations, $k=0,1$.
We always have $s_k^{(0)}=0$ due to the initialization of $s_k$.
See Figure \ref{algRunningEx}.

Assume that $s_0$ is updated in line \ref{updateS} in iteration $i$.
The point $\pp 0{s_0^{(i)}}$ is in the half-plane
$\LHP(\pp {1}{s_{1}^{(i-1)}},\pp 0{s_0^{(i-1)}})$,
but not on the line $\LLL(\pp {1}{s_{1}^{(i-1)}},\pp 0{s_0^{(i-1)}})$. Therefore, we
have the following observation.

\begin{observation}\label{rotation}
When $s_k$ is updated, the temporary line is rotated counterclockwise around $s_{1-k}$ by an angle
less than $\pi$.
\end{observation}

Assume in the following
that the convex hulls of $P_0$ and $P_1$ are disjoint so that separating common tangents exist.
Let $(r_0,r_1)$ be the indices that define the separating common tangent such that
$P_k$ is contained in $\RHP(\pp{1-k}{r_{1-k}},\pp k{r_k})$, i.e.,
$(r_0,r_1)$ is the result we are going to prove that the algorithm returns.

Since $\CH(P_k)$ is convex, the temporary line always divides
$\CH(P_k)$ into two convex parts.
If we follow the temporary line from
$\pp {1-k}{s_{1-k}}$ in the direction towards
$\pp k{s_k}$, we enter $\CH(P_k)$ at some point $x$
and thereafter leave $\CH(P_k)$ again at some
point $y$.
We clearly have $x=y$ if and only if the temporary line is a tangent to $\CH(P_k)$,
since if $x=y$ and the line was no tangent,
$\CH(P_k)$ would only be a line segment.
The part of the boundary of $\CH(P_k)$ counterclockwise from $x$ to $y$ is in
$\RHP(\pp {1-k}{s_{1-k}},\pp k{s_k})$ whereas the part from $y$ to $x$ is on
$\LHP(\pp {1-k}{s_{1-k}},\pp k{s_k})$. We therefore have the following observation.

\begin{observation}\label{kalotten}
Let $d$ be the index of the corner of $\CH(P_k)$ strictly after $y$ in counterclockwise order.
There exists a corner $\pp kt$ of $P_k$ such that $\TTT(\pp {1-k}{s_{1-k}},\pp k{s_k},\pp kt)=1$
if and only if $\TTT(\pp {1-k}{s_{1-k}},\pp k{s_k},\pp kd)=1$.
\end{observation}

Let $c_k$ be the index of the first corner of $\CH(P_k)$
when following $\CH(P_k)$ in counterclockwise order from $y$,
$c_k=0,\ldots,n_k-1$. If $y$ is itself a corner
of $\CH(P_k)$, we have $\pp k{c_k}=y$.
By Observation \ref{kalotten} we see that
$\TTT(\pp {1-k}{s_{1-k}},\pp k{s_k},\pp k{c_k})\geq 0$ with equality if and only if
$\pp k{c_k}=\pp k{s_k}=y$.
Let $c_k^{(0)}$ be $c_k$ when only line \ref{initialize} has been executed.
Consider now the value of $c_k$ after $i=1,2,\ldots$ iterations.
Let $c_k^{(i)}=c_k$ and add $n_k$ to $c_k^{(i)}$
until $c_k^{(i)}\geq c_k^{(i-1)}$.
This gives a non-decreasing sequence of indices $c_k^{(0)},c_k^{(1)},\ldots$
of the first corner of $\CH(P_k)$ in $\LHP(\pp{1-k}{s_{1-k}},\pp k{s_k})$.
Actually, we prove in the following that we need to add $n_k$ to $c_k^{(i)}$ at most once
before $c_k^{(i)}\geq c_k^{(i-1)}$.
If $r_k<c_k^{(0)}$ we add $n_k$ to $r_k$. Thus we have
$0=s_k^{(0)}\leq c_k^{(0)}\leq r_k<2n_k$.

The following lemma intuitively says that the algorithm does not ``jump over'' the correct
solution and it expresses the main idea in our proof of correctness.

\begin{lemma}\label{mainLemma}
After each iteration $i=0,1,\ldots$ and for each $k=0,1$ we have
$$0\leq s_k^{(i)}\leq c_k^{(i)}\leq r_k < 2n_k.$$
Furthermore, the test in line \ref{testStop} is never positive.
\end{lemma}

\begin{proof}
We prove the lemma for $k=0$.
From the definition of $r_0$, we get that
$0=s_0^{(0)}\leq c_0^{(0)}\leq r_0 < 2n_0$.
Since the sequence $s_0^{(0)},s_0^{(1)},\ldots$ is non-decreasing, the
inequality $0\leq s_k^{(i)}$ is true for every $i$.

Now, assume inductively that $s_0^{(i-1)}\leq c_0^{(i-1)}\leq r_0$ and consider what happens during
iteration $i$. If neither $s_0$ nor $s_1$ is updated,
the statement is trivially true from the induction
hypothesis, so assume that an update happens.

By the \emph{old temporary line} we mean the temporary line defined by
$(s_0^{(i-1)},s_1^{(i-1)})$ and the \emph{new temporary line} is the one
defined by $(s_0^{(i)},s_1^{(i)})$.
The old temporary line enters $\CH(P_0)$ at some point $x$ and exits at some point
$y$ when followed from $\pp 1{s_1^{(i-1)}}$. Likewise,
let $v$ be the point where the new temporary line exits $\CH(P_0)$ when followed from
$\pp 1{s_1^{(i)}}$. The point $x$ exists since the convex hulls are disjoint.

Assume first that the variable $u$ in the algorithm is $0$,
i.e., a corner of the polygon $P_0$ is traversed. In this case
$s_{1}^{(i-1)}=s_{1}^{(i)}$.

\begin{figure}
\centering
\begin{tikzpicture}[scale=0.2]

\input{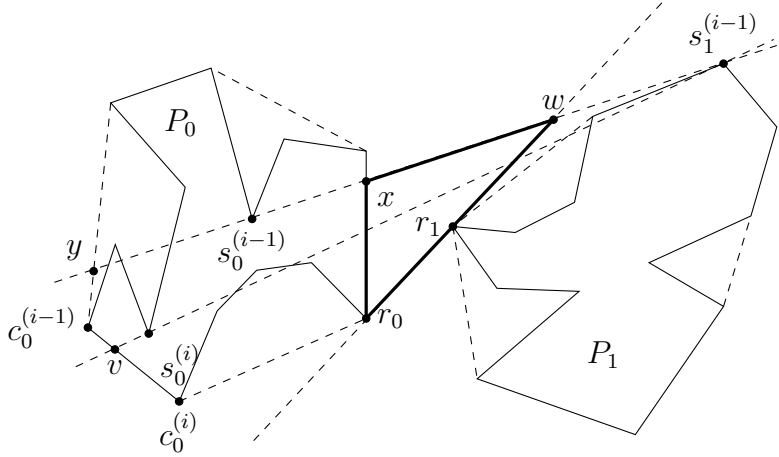}

\draw [dashed] (Z) -- (C);
\draw [dashed] (D) -- (H);
\draw [dashed] (I) -- (Y);

\draw [dashed] (W) -- (M);
\draw [dashed] (O) -- (Q);
\draw [dashed] (T) -- (W);

\filldraw [black] (Y) circle (7pt);
\node at (Y) [right] {$r_0$};

\filldraw [black] (W) circle (7pt);
\node at (W) [left] {$r_1$};

\filldraw [black] (B) circle (7pt);
\node at (B) [below] {$s_0^{(i-1)}$};

\filldraw [black] (S) circle (7pt);
\node at (S) [above] {$s_1^{(i-1)}$};

\filldraw [black] (H) circle (7pt);
\node at (H) [left] {$c_0^{(i-1)}$};

\draw (intersection of S--B and Y--Z) coordinate (XX);
\filldraw [black] (XX) circle (7pt);
\node at (XX) [anchor = north west] {$x$};

\draw (intersection of S--B and D--H) coordinate (YY);
\filldraw [black] (YY) circle (7pt);
\node at (YY) [anchor = south east] {$y$};

\draw [dashed] (-1.1, 26.8) -- (46.4, 42.5);

\filldraw [black] (F) circle (7pt);
\node at (F) [anchor = north west] {$s_0^{(i)}$};

\filldraw [black] (I) circle (7pt);
\node at (I) [below] {$c_0^{(i)}$};

\draw (intersection of S--F and H--I) coordinate (VV);
\filldraw [black] (VV) circle (7pt);
\node at (VV) [below] {$v$};

\draw [dashed] (0.3, 21.2) -- (46.2, 42.9);




\draw [dashed] (12.0, 16.3) -- (39.0, 45.5);

\draw (intersection of B--S and Y--W) coordinate (WW);
\filldraw [black] (WW) circle (7pt);
\node at (WW) [above] {$w$};
\draw [very thick] (XX) -- (WW) -- (Y) -- cycle;

\end{tikzpicture}
\caption{An update of $s_0$ happens in iteration $i$ from $s_0^{(i-1)}$ to $s_0^{(i)}$
and $\pp 0{c_0}$ moves forward on $\CH(P_0)$ from $\pp 0{c_0^{(i-1)}}$ to $\pp 0{c_0^{(i)}}$.
The relevant corners are marked and labeled with their indices.
The polygon $\mathcal C$ from the proof of Lemma \ref{mainLemma} is
drawn with thick lines.}
\label{sUpdate}
\end{figure}

We now prove
$s_0^{(i)}\leq c_0^{(i)}$. Assume that 
$\pp 0{s_0^{(i-1)}}\neq \pp 0{c_0^{(i-1)}}$. The situation is depicted in Figure \ref{sUpdate}.
In this case
$\TTT(\pp {1}{s_{1}^{(i-1)}},\pp 0{s_0^{(i-1)}},\pp 0{c_0^{(i-1)}})=1$.
Hence, the update happens when $\pp 0{c_0^{(i-1)}}$ is traversed or
earlier, so $s_0^{(i)}\leq c_0^{(i-1)}\leq c_0^{(i)}$.
Assume now that $\pp 0{s_0^{(i-1)}}=\pp 0{c_0^{(i-1)}}$.
We cannot have $c_0^{(i)}=c_0^{(i-1)}$ since
$\TTT(\pp 1{s_1^{(i)}},\pp 0{s_0^{(i)}},\pp 0{c_0^{(i-1)}})=
-\TTT(\pp 1{s_1^{(i-1)}},\pp 0{s_0^{(i-1)}},\pp 0{s_0^{(i)}})=-1$, therefore
$c_0^{(i)}>c_0^{(i-1)}$.
Consider the corner $\pp 0{c'}$ on $\CH(P_0)$
following $\pp 0{c_0^{(i-1)}}$ in counterclockwise order, $c'>c_0^{(i-1)}$.
Due to the minimality of $c'$, we have $c'\leq c_0^{(i)}$.
By Observation \ref{kalotten},
$\TTT(\pp {1}{s_{1}^{(i-1)}},\pp 0{s_0^{(i-1)}},\pp 0{c'})=1$.
Therefore, $s_0$ must be updated when $\pp 0{c'}$ is traversed or earlier,
so $s_0^{(i)}\leq c'\leq c_0^{(i)}$.

For the inequality $c_0^{(i)}\leq r_0$,
consider the new temporary line in the direction from
$\pp {1}{s_{1}^{(i-1)}}$ to $\pp 0{s_0^{(i)}}$.
We prove that $v$ is in the part of $\CH(P_0)$
from $y$ counterclockwise to $r_0$.
The point $\pp 0{s_0^{(i)}}$ is in the polygon $Q$ defined by the segment $xy$ together with
the part of $\CH(P_0)$ from $y$ counterclockwise to $x$. Therefore, the new temporary
line enters and exits $Q$. It cannot exit through the segment $xy$, since the old and
new temporary lines intersect at $\pp {1}{s_{1}^{(i-1)}}$, which is in
$\CH(P_{1})$. Therefore, $v$ must be on the part of $\CH(P_0)$ from $y$ to $x$.
If $r_0$ is on the part of $\CH(P_0)$ from $x$ counterclockwise to $y$,
then $v$ is on the part from $y$ to $r_0$ as we wanted.

Otherwise, assume for contradiction that the points appear in the order
$y$, $\pp 0{r_0}$, $v$, $x$ counterclockwise along $\CH(P_0)$, where $\pp 0{r_0}\neq v\neq x$.
The endpoints of the segment $\pp {1}{s_{1}^{(i-1)}}x$ are on different
sides of the tangent defined by $(r_0,r_1)$, so the segment intersects the tangent at a point $w$.
The part of $\CH(P_0)$ from $\pp 0{r_0}$ to $x$
and the segments $xw$ and $w\pp 0{r_0}$ form a simple polygon
$\mathcal C$, see Figure \ref{sUpdate} for an example.
The new temporary line enters $\mathcal C$ at
the point $v$, so it must leave $\mathcal C$ after $v$.
The line cannot cross $\CH(P_0)$ after $v$ since $\CH(P_0)$ is convex.
It also cannot cross the segment $xw$ at a point after $v$ since the old and the new temporary
line cross before $v$, namely at $\pp {1}{s_{1}^{(i-1)}}$.
The tangent defined by $(r_0,r_1)$ and the new temporary line
intersect before $v$ since the endpoints of the segment
$\pp {1}{s_{1}^{(i-1)}}v$ are on different sides of the tangent.
Therefore, the line cannot cross the segment $w\pp 0{r_0}$ at a point after $v$.
Hence, the line cannot exit $\mathcal C$. That is a contradiction.

Therefore, $v$ is on the part of $\CH(P_0)$ from
$y$ to $\pp 0{r_0}$ and hence the first corner $\pp 0{c_0^{(i)}}$
of $\CH(P_0)$ after $v$
must be before or coincident with $\pp 0{r_0}$, so that $c_0^{(i)}\leq r_0$.

Assume now that $u=1$ in the beginning of iteration $i$, i.e.,
a corner of the other polygon $P_{1}$ is traversed.
In that case, we have $s_0^{(i)}=s_0^{(i-1)}\leq c_0^{(i-1)}\leq c_0^{(i)}$, and
we need only prove $c_0^{(i)}\leq r_0$.
Observation \ref{rotation} gives that
$v$ is in the part of $\CH(P_0)$ from $y$ to $x$, since the new temporary line is obtained
by rotating the old temporary line
counterclockwise around $\pp 0{s_0^{(i-1)}}$ by an angle less than $\pi$.
That $v$ appears before $\pp 0{r_0}$ on $\CH(P_0)$ counterclockwise from $y$ follows
from exactly the same arguments as in the case $u=0$.

We have nowhere used the test at line \ref{testStop} to conclude that $s_k<2n_k$.
Hence, the test is never positive. This completes the proof.
\end{proof}

We are now ready to prove that Algorithm \ref{sepTanAlg} has the desired properties.

\begin{theorem}\label{mainThm}
If the polygons $P_0$ and $P_1$ have separating common tangents, Algorithm \ref{sepTanAlg}
returns a pair of indices $(s_0,s_1)$ defining a separating common tangent
such that $P_k$ is contained in
$\RHP(\pp {1-k}{s_{1-k}},\pp k{s_k})$ for $k=0,1$. If no separating common tangents exist,
the algorithm returns $\texttt{NULL}$. The algorithm runs in linear time and uses
constant workspace.
\end{theorem}

\begin{proof}
Assume first that the algorithm returns
$(s_0,s_1)$. We know that $s_k<2n_k$ for each $k=0,1$,
since we never update $s_k$ to values as large as $2n_k$.
Therefore, we have that
$\pp k{t}\in\RHP(\pp {1-k}{s_{1-k}},\pp k{s_k})$ for each
$k=0,1$ and each
$t\in\{2n_k,\ldots,3n_k-1\}$. Hence the pair $(s_0,s_1)$ indeed defines
the separating common tangent.

Assume now that there exists a separating common tangent.
By Lemma \ref{mainLemma}, a pair $(s_0,s_1)$ is returned. As we already saw,
this means that $(s_0,s_1)$ defines the separating common tangent.

If an update happens in iteration $i$,
the sum $s_0+s_1$ is increased by at least
$\frac{i-j}2$, where $j\geq 0$ was the previous iteration where an update happened.
Inductively, we see that when the final update of $s_0$ and $s_1$ happens,
there has been at most $2(s_0+s_1)$ iterations.
After the final update, at most $3n_0-s_0+3n_1-s_1$ iterations follow. In total,
the algorithm performs $3n_0+s_0+3n_1+s_1\leq 5(n_0+n_1)$ iterations.
\end{proof}

\section{Computing outer common tangents}

In this section, we assume that two polygons $P_0$ and $P_1$ are given such that
their convex hulls are disjoint. We assume that
the corners $\pp 00,\ldots,\pp 0{n_0-1}$ of $P_0$ are given in counterclockwise order
and the corners $\pp 10,\ldots,\pp 1{n_1-1}$ of $P_1$ are given in clockwise order.
We say that the \emph{orientation} of $P_0$ and $P_1$ is counterclockwise and clockwise,
respectively. We prove that
Algorithm \ref{outTanAlg}
returns two indices $(s_0,s_1)$ that define an outer common tangent
such that $P_0$ and $P_1$ are both
contained in $\RHP(\pp 0{s_0},\pp 1{s_1})$.

\begin{algorithm}[h]
\LinesNumbered
\DontPrintSemicolon
\SetArgSty{}
\SetKwInput{Input}{Input}
\SetKwInput{Output}{Output}
\SetKw{Report}{report}
\SetKwIF{If}{ElseIf}{Else}{if}{}{else if}{else}{end if}
\SetKwFor{Foreach}{for each}{}{end for}
\SetKwFor{While}{while}{}{end while}
$s_0\gets 0$;\quad$t_0\gets 1$;\quad$s_1\gets 0$;\quad$t_1\gets 1$;\quad$u\gets 0$\;
\nllabel{initialize2}
\While{$t_0<2n_0$ or $t_1<2n_1$} {\nllabel{while2}
  \If {$\TTT(\pp {0}{s_{0}},\pp 1{s_1},\pp u{t_u})=1$} {\nllabel{testSide2}
    $s_u\gets t_u$\;
    $t_{1-u}\gets s_{1-u}+1$\;
  }
  $t_u\gets t_u+1$\;
  $u\gets 1-u$\;
}
\Return {$(s_0,s_1)$}
\caption{$\ttt{OuterCommonTangent}(P_0,P_1)$}
\label{outTanAlg}
\end{algorithm}

\begin{figure}
\centering
\begin{tikzpicture}[scale=0.2]

\input{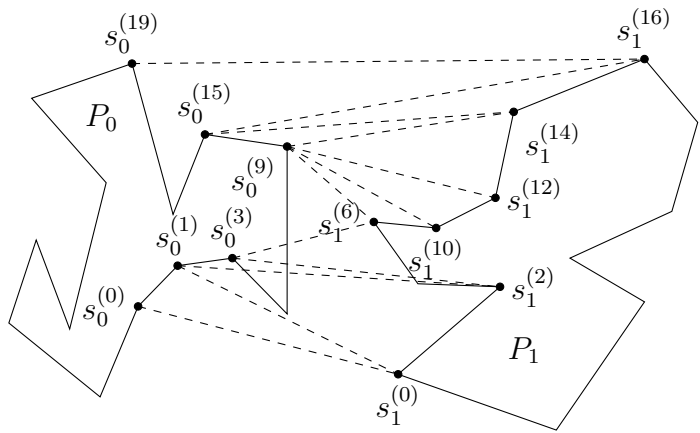}

\draw [dashed] (J) -- (M);
\filldraw [black] (J) circle (7pt);
\node at (J) [left] {$s_0^{(0)}$};
\filldraw [black] (M) circle (7pt);
\node at (M) [below] {$s_1^{(0)}$};

\draw [dashed] (K) -- (M);
\filldraw [black] (K) circle (7pt);
\node at (K) [above] {$s_0^{(1)}$};

\draw [dashed] (K) -- (Y1);
\filldraw [black] (Y1) circle (7pt);
\node at (Y1) [right] {$s_1^{(2)}$};

\draw [dashed] (L) -- (Y1);
\filldraw [black] (L) circle (7pt);
\node at (L) [above] {$s_0^{(3)}$};

\draw [dashed] (L) -- (W);
\filldraw [black] (W) circle (7pt);
\node at (W) [left] {$s_1^{(6)}$};

\draw [dashed] (Z) -- (W);
\filldraw [black] (Z) circle (7pt);
\node at (Z) [anchor = north east] {$s_0^{(9)}$};

\draw [dashed] (Z) -- (V);
\filldraw [black] (V) circle (7pt);
\node at (V) [below] {$s_1^{(10)}$};

\draw [dashed] (Z) -- (U);
\filldraw [black] (U) circle (7pt);
\node at (U) [right] {$s_1^{(12)}$};

\draw [dashed] (Z) -- (T);
\filldraw [black] (T) circle (7pt);
\node at (T) [anchor = north west] {$s_1^{(14)}$};

\draw [dashed] (A) -- (T);
\filldraw [black] (A) circle (7pt);
\node at (A) [above] {$s_0^{(15)}$};

\draw [dashed] (A) -- (S);
\filldraw [black] (S) circle (7pt);
\node at (S) [above] {$s_1^{(16)}$};

\draw [dashed] (C) -- (S);
\filldraw [black] (C) circle (7pt);
\node at (C) [above] {$s_0^{(19)}$};



\end{tikzpicture}
\caption{Algorithm \ref{outTanAlg} running on two polygons $P_0$ and $P_1$.
The corners $\pp k{s_k^{(i)}}$ are marked and labeled as
$s_k^{(i)}$ for the initial values
$s_k^{(0)}$ and after each iteration $i$ where an update of $s_k$ happens.
The segments $\pp 0{s_0^{(i)}}\pp 1{s_1^{(i)}}$
on the temporary line are dashed.}
\label{algOuterRunningEx}
\end{figure}

As in the case of separating common tangents,
we define $s_k^{(i)}$ as the value of $s_k$ after $i=0,1,\ldots$ iterations of
the loop at line \ref{while2} of Algorithm \ref{outTanAlg}. See Figure \ref{algOuterRunningEx}.
For this algorithm, we get a slightly different analogue to Observation \ref{rotation}:

\begin{observation}\label{rotation2}
When $s_k$ is updated, the temporary line is rotated around $s_{1-k}$
in the orientation of $P_{1-k}$ by an angle less than $\pi$.
\end{observation}

Let $y$ be the point where the temporary line enters $\CH(P_k)$ when followed from
$\pp{1-k}{s_{1-k}}$ and $x$ the point where it exits $\CH(P_k)$.
We have the following analogue of Observation \ref{kalotten}.

\begin{observation}\label{kalotten2}
Let $d$ be the index of the corner of $\CH(P_k)$ strictly after $y$ following the orientation of $P_k$.
There exists a corner $\pp kt$ of $P_k$ such that $\TTT(\pp 0{s_0},\pp 1{s_1},\pp kt)=1$
if and only if $\TTT(\pp 0{s_0},\pp 1{s_1},\pp kd)=1$.
\end{observation}

Let $c_k$ be the index of the first corner of $\CH(P_k)$ after $y$
following the orientation of $P_k$, where $\pp k{c_k}=y$ if $y$ is itself a corner of $\CH(P_k)$.
By Observation \ref{kalotten2},
we have $\TTT(\pp 0{s_0},\pp 1{s_1},\pp k{c_k})\geq 0$ with equality if and only if
$\pp k{c_k}=\pp k{s_k}=y$.
Define a non-decreasing sequence $c_k^{(0)},c_k^{(1)},\ldots$ of the value of
$c_k$ after $i=0,1,\ldots $ iterations as we did for separating tangents. Also, let the indices
$(r_0,r_1)$ define the outer common tangent that we want the algorithm to return
such that $c_k^{(0)}\leq r_k<2n_k$. We can now state the analogue to Lemma \ref{mainLemma}
for outer common tangents.

\begin{lemma}\label{mainLemma2}
After each iteration $i=0,1,\ldots$ and for each $k=0,1$ we have
$$0\leq s_k^{(i)}\leq c_k^{(i)}\leq r_k < 2n_k.$$
\end{lemma}

\begin{proof}
Assume $k=0$ and the induction hypothesis $s_0^{(i-1)}\leq c_0^{(i-1)}\leq r_0$.
The inequality $s_0^{(i)}\leq c_0^{(i)}$ can be proven exactly as in the proof of Lemma
\ref{mainLemma}. Therefore, consider the inequality $c_0^{(i)}\leq r_0$ and assume that
an update happens in iteration $i$.

Let the \emph{old temporary line}
and the \emph{new temporary line} be the lines defined by the indices
$(s_0^{(i-1)},s_1^{(i-1)})$ and $(s_0^{(i)},s_1^{(i)})$, respectively.
Let $y$ and $x$ be the points where the old temporary line enters and exits
$\CH(P_0)$ followed from $\pp 1{s_1^{(i-1)}}$, respectively,
and let $v$ be the point where the new temporary line enters
$\CH(P_0)$. The points $y$ and $v$ exist since the convex hulls of $P_0$ and
$P_1$ are disjoint.

\begin{figure}
\centering
\begin{tikzpicture}[scale=0.2]

\input{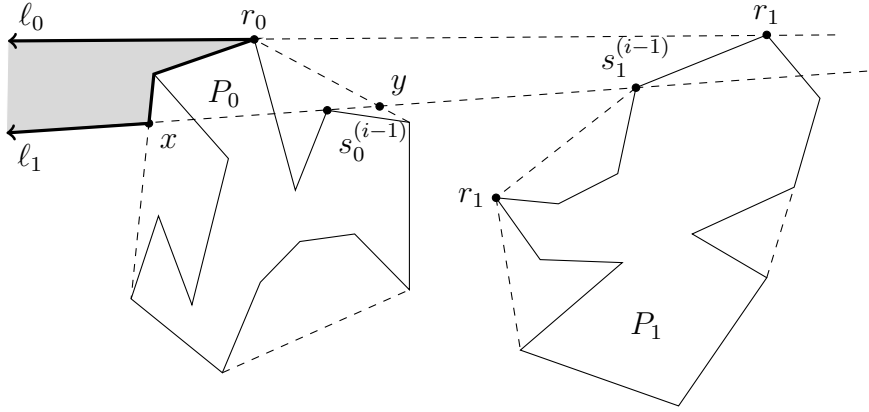}

\draw [dashed] (Z) -- (C);
\draw [dashed] (D) -- (H);
\draw [dashed] (I) -- (Y);

\draw [dashed] (W) -- (M);
\draw [dashed] (O) -- (Q);
\draw [dashed] (T) -- (W);


\filldraw [black] (C) circle (7pt);
\node at (C) [above] {$r_0$};

\filldraw [black] (S) circle (7pt);
\node at (S) [above] {$r_1$};

\filldraw [black] (W) circle (7pt);
\node at (W) [left] {$r_1$};

\filldraw [black] (A) circle (7pt);
\node at (A) [anchor = north west] {$s_0^{(i-1)}$};

\filldraw [black] (T) circle (7pt);
\node at (T) [above] {$s_1^{(i-1)}$};

\draw (intersection of A--T and D--H) coordinate (XX);
\filldraw [black] (XX) circle (7pt);
\node at (XX) [anchor = north west] {$x$};

\draw (intersection of A--T and C--Z) coordinate (YY);
\filldraw [black] (YY) circle (7pt);
\node at (YY) [anchor = south west] {$y$};

\draw [white, fill=gray!30] (-7.0, 40.9) coordinate (LL0) --
(C) -- (D) -- (XX) -- (-7.1, 34.8) coordinate (LL1) -- cycle;

\draw [<->, very thick] (LL0) --
(C) -- (D) -- (XX) -- (LL1);

\draw [dashed] (47.4, 41.3) -- (LL0);
\draw [dashed] (49.5, 38.9) -- (LL1);

\node at (LL0) [anchor = south west] {$\ell_0$};
\node at (LL1) [anchor = north west] {$\ell_1$};

\end{tikzpicture}
\caption{The area $\mathcal A$
from the proof of Lemma \ref{mainLemma2} in grey. The relevant corners are
marked and labeled with their indices.}
\label{outTanArea}
\end{figure}

Assume first that
the variable $u$ in the algorithm equals $0$ when the update happens. We prove that $v$ is in the part of $\CH(P_0)$ from
$y$ to $\pp 0{r_0}$ following the orientation of $P_0$, which is counterclockwise.
The point $\pp 0{s_0^{(i)}}$ is in the simple polygon $Q$ bounded the part of $\CH(P_0)$ from
$y$ counterclockwise to $x$ and the segment $xy$. Therefore, the new temporary line must enter $Q$
to get to $\pp 0{s_0^{(i)}}$. It cannot enter through $xy$, since
the old and new temporary line cross at $\pp 1{s_1^{(i-1)}}$ which is not in $\CH(P_k)$ by assumption.
Therefore, it must enter through the part of $\CH(P_0)$ from $y$ to $x$, so $v$ is in this part.
If $r_0$ is not in the part of $\CH(P_0)$ from $y$ to $x$, it is clearly true that
$v$ is in the part from $y$ to $\pp 0{r_0}$.
Otherwise, assume for contradiction that the points appear on $\CH(P_0)$ in the order
$y,\pp 0{r_0},v,x$ and $\pp 0{r_0}\neq v\neq x$.
Let $\ell_0$ be the half-line starting at $\pp 0{r_0}$ following the tangent away from
$\pp 1{r_1}$, and let $\ell_1$ be the half-line starting at $x$
following the old temporary line away from $\pp 1{s_1^{(i-1)}}$.
The part of $\CH(P_0)$ from $\pp 0{r_0}$ to $x$ and the half-lines
$\ell_0$ and $\ell_1$ define a possibly unbounded area $\mathcal A$ outside $\CH(P_0)$,
see Figure \ref{outTanArea}.
We follow the new temporary line from $\pp 1{s_1^{(i-1)}}$ towards $v$.
The point $\pp 1{s_1^{(i-1)}}$ is not in $\mathcal A$ and
the new temporary line exits $\mathcal A$ at $v$ since it enters $\CH(P_0)$ at $v$,
so it must enter $\mathcal A$ somewhere
at a point on the segment $\pp 1{s_1^{(i-1)}}v$.
It cannot enter through $\CH(P_0)$ since $\CH(P_0)$ is convex.
It cannot enter through $\ell_0$ since
$v$ and $\pp 1{s_1^{(i-1)}}$ are on the same side of the outer common tangent.
It cannot enter through $\ell_1$ since the old and new temporary line
intersect in $\pp 1{s_1^{(i-1)}}$, which is not in $\mathcal A$. That is a contradiction,
so $v$ is on the part of $\CH(P_0)$ from $y$ to $\pp 0{r_0}$.
Hence, the first corner after $y$ is coincident with or before $\pp 0{r_1}$, i.e.,
$c_0^{(i)}\leq r_0$.

Assume now that $u=1$ in the beginning of iteration $i$ so that a corner of the polygon $P_1$ is traversed.
Observation \ref{rotation2} gives that $v$ is on the part of $\CH(P_0)$ from $y$ counterclockwise to $x$.
It follows that $v$ appears before $\pp 0{r_0}$ on
$\CH(P_0)$ counterclockwise from $y$ from exactly the same arguments as in the case $u=0$.
\end{proof}

\begin{lemma}\label{update2}
If $\pp 0{s_0}\neq\pp 0{r_0}$ or $\pp 1{s_1}\neq\pp 1{r_1}$, 
then $\TTT(\pp 0{s_0},\pp 1{s_1},\pp k{t})=1$ for some
$k=0,1$ and some index $t\in\{s_k+1,\ldots,r_k\}$.
\end{lemma}

\begin{proof}
Assume that $\TTT(\pp 0{s_0},\pp 1{s_1},\pp k{r_k})\leq 0$
for $k=0,1$, since otherwise, we are done.
Likewise, assume that all of the part
$P_0[s_0,r_0]$ of
$P_0$ from $\pp 0{s_0}$ to $\pp 0{r_0}$
is in
$\RHP(\pp 0{s_0},\pp 1{s_1})$. The part
$P_0[s_0,r_0]$ separates $\pp 1{s_1}$ from
$\pp 1{r_1}$ in the set
$W=\RHP(\pp 0{s_0},\pp 1{s_1})\cap\RHP(\pp 0{r_0},\pp 1{r_1})$.
Since the part
$P_1[s_1,r_1]$ of $P_1$ from $\pp 1{s_1}$ to $\pp 1{r_1}$ cannot cross
$P_0[s_0,r_0]$
or $\LLL(\pp 0{r_0},\pp 1{r_1})$, it must exit
and enter
$W$ through points on $\LLL(\pp 0{s_0},\pp 1{s_1})$ when
followed from $\pp 1{s_1}$, and hence the claim is true.
\end{proof}

We can now prove the stated properties of Algorithm \ref{outTanAlg}.

\begin{theorem}
If the polygons $P_0$ and $P_1$ have disjoint convex hulls, Algorithm \ref{outTanAlg}
returns a pair of indices $(s_0,s_1)$ defining an outer common tangent
such that $P_0$ and $P_1$ are contained in
$\RHP(s_0,s_1)$. The algorithm runs in linear time and uses
constant workspace.
\end{theorem}

\begin{proof}
Assume that the pair $(s_0,s_1)$ does not define the outer common tangent.
By Lemma \ref{update2}, an update of $s_0$ or $s_1$ happens when
$\pp 0{r_0}$ or $\pp 1{r_1}$ is traversed or before. By Lemma \ref{mainLemma2},
the algorithm does not terminate before
$\pp 0{r_0}$ and $\pp 1{r_1}$ has been traversed. Hence, when the algorithm
terminates, $(s_0,s_1)$ defines the outer common tangent.

Like in the proof of Theorem \ref{mainThm}, we see
inductively that when the final update of $s_0$ and $s_1$ happens,
there has been at most $2(s_0+s_1)$ iterations.
After that, at most $2n_0-s_0+2n_1-s_1$ iterations follow. Hence,
the algorithm terminates after at most $4n_0+4n_1$ iterations.
\end{proof}

\section{Concluding Remarks}

\begin{figure}
\centering
\begin{tikzpicture}[scale=0.5]

\draw (3.7,10.5) coordinate (A)
-- (3.7, 3.7) coordinate (B)
-- (12.1, 3.9) coordinate (C)
-- (12.0, 5.0) coordinate (D)
-- (11.0, 5.7) coordinate (E)
-- (5.6, 5.2) coordinate (F)
-- cycle;

\draw (6.7, 7.8) coordinate (G)
-- (6.0, 7.3) coordinate (H)
-- (5.2, 7.3) coordinate (I)
-- (7.5, 11.2) coordinate (J)
-- (7.5, 5.9) coordinate (K)
-- cycle;

\draw [dashed] (2.72, 10.32) -- (8.41, 11.39);

\node at (4.6, 4.7) {$P_0$};
\node at (8, 9.3) {$P_1$};

\filldraw [black] (A) circle (3pt);
\node at (A) [above] {$s_0^{(0)}$};
\filldraw [black] (G) circle (3pt);
\node at (G) [above] {$s_1^{(0)}$};

\end{tikzpicture}
\caption{Two polygons $P_0$ and $P_1$ where Algorithm \ref{outTanAlg} does not
work for the initial values of $s_0$ and $s_1$ as shown. The correct tangent is drawn as
a dashed line.}
\label{infiniteLoop}
\end{figure}
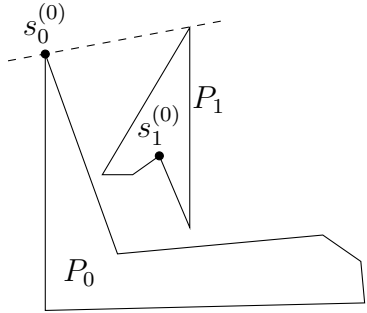

We have described an algorithm for computing the separating common tangents of two
simple polygons in linear time using constant workspace.
We have also described an algorithm for computing outer common
tangents using linear time and constant workspace when the convex hulls of the polygons are disjoint.
Figure \ref{infiniteLoop}
shows an example where Algorithm \ref{outTanAlg} does not work when applied to
two disjoint polygons with overlapping convex hulls. In fact,
if there was no bound on the values $t_0$ and $t_1$ in the loop at line \ref{while2},
the algorithm would update $s_0$ and $s_1$ infinitely often and never find the correct
tangent.
An obvious improvement
is to find an equally fast and space efficient algorithm which does not require the convex
hulls to be disjoint.
An algorithm for computing an outer common tangent of two polygons, when such one exists, also
decides if one convex hull is completely contained in the other. Together with the algorithm for
separating common tangents presented in Section \ref{sepTan}, we would have an optimal
algorithm for deciding the complete relationship between the convex hulls:
if one is contained in the other, and if not, whether they are disjoint or not.
However, keeping in mind that it is harder to compute an outer common tangent of intersecting
convex polygons than of disjoint ones \cite{guibas1991}, it would not be surprising if it was also
harder to compute an outer common tangent of general simple polygons than
simple polygons with disjoint convex hulls when only constant workspace is available.





\section*{Acknowledgments}

We would like to thank Mathias Tejs B\ae{}k Knudsen for pointing out the error
in the algorithm for separating common tangents in the preliminary version of the paper
\cite{abrahamsen}.

\bibliography{proj}{}
\bibliographystyle{plain}

\end{document}